\documentclass[a4paper,12pt,thmsa]{article}
\usepackage[a4paper,marginratio={1:1},scale={0.72,0.74},footskip=7mm,headsep=10mm]{geometry}
\usepackage{amsmath,amsfonts,amssymb,amsthm}

\usepackage{amsfonts}
\usepackage{amssymb,amsmath,latexsym}
\usepackage{graphicx}
\usepackage{pstricks,pstricks-add}

\newtheorem{thm}{Theorem} 
\newtheorem{lemma}{Lemma}

\newtheorem{remark}{Remark}

\newcommand{\eps}{\varepsilon}
\renewcommand{\P}{\mathbb{P}}
\newcommand{\var}{\mbox{\rm Var}}

\newcommand{\E}{\mathbb{E}}

\def\p{\mathbb{P}}

\parskip \medskipamount

\begin{document}

\title{Reconstructing pedigrees using probabilistic analysis of ISSR amplification.
\thanks{This work was supported by MODEMAVE research project from the R\'egion Pays de la Loire. Acces to molecular data
was supported by both EUROGENI project, funded by R\'egion Pays de la Loire (dynamiques de fili\`ere) and by BRIO project 
funded by same R\'egion Pays de la Loire and the Fonds Unique Interminist\'eriel.}}

\author{Lo\"ic Chaumont$^1$, Val\'ery Mal\'ecot$^2$, Richard Pymar$^3$ and Chaker Sbai$^{4}$}

\date{\today}

\maketitle

\noindent {\small $^1$ LAREMA -- UMR CNRS 6093, Universit\'e d'Angers, 2 bd Lavoisier, 49045 Angers Cedex 01\\
%\email{loic.chaumont@univ-angers.fr}
$^2$ IRHS -- UMR 1345, Agrocampus Ouest Angers, 2 rue Le N\^otre, 49045 Angers Cedex 01\\
%\email{valery.malecot@agrocampus-ouest.fr}
$^3$ Department of Mathematics -- University College London, Gower Street, London WC1E 6BT}\\
%\email{r.pymar@ucl.ac.uk}
$^4$ PEGASE -- UMR 1348, Agrocampus Ouest Rennes, 65 rue de Saint-Brieuc, CS 84215, 35042 Rennes Cedex 
%\email{sbaichaker@yahoo.fr}

\begin{abstract}
Data obtained from ISSR amplification may readily be extracted but only allows us to know, for each gene, if a specific 
allele is present or not. From this 
partial information we provide a probabilistic method to reconstruct the pedigree corresponding to some families of diploid cultivars.
This method consists in determining for each individual what is the most likely couple of parent pair amongst all older individuals, 
according to some probability measure. The construction of this measure bears on the fact that the probability to observe the specific
alleles in the child, given the status of the parents does not depend on the generation and is the same for each gene. This assumption
is then justified from a convergence result of gene frequencies which is proved here. Our reconstruction method is applied to a 
family of 85 living accessions representing the common broom {\it Cytisus scoparius}. 
\end{abstract}

\noindent{\it Mathematics Subject Classification $(2000)$}: 92D25; 92D10; 60F15

\section{Introduction}\label{in}

A pedigree is a graph such that each vertex has indegree equal to 0 or 2 and any outdegree. When it represents family relationships 
between living individuals, edges are directed from parents to children. By reconstruction of the pedigree of a family of 
some set of individuals, we mean a way to determine the most likely pedigree relating theses individuals given some information such 
as phenotype, genotype, date of birth, data obtained from professional breeders,... It may happen that this information is known only for a part of 
the population or even that 
the number of missing individuals is unknown. To each situation corresponds some specific methods. Deterministic methods based on
the maximum parsimony principle and using purely combinatorial arguments allow us to reconstruct the minimal pedigree relating 
individuals in accordance with their types, see Chapter 4 in \cite{ss}, \cite{sh} or \cite{bl}. There are also numerous different stochastic 
methods of reconstruction of pedigrees, see for instance \cite{klkh}, \cite{th}, \cite{bs}, \cite{bl}.
In any case, the method consists in finding a 'nice' probabilistic framework in which we may find the most likely pedigree relating some 
set of individuals. Some models focus on the reconstruction of the lineages by estimating transition probabilities between nodes.
Reconstructing the pedigree then comes down to the construction of a Markov chain. This method is quite popular when making use 
of identity by descent (IBD) data, \cite{klkh}. In this case, a statistical inference based on Monte Carlo Markov chains and Bayesian 
statistics are used to infer transition probabilities between nodes of the graph, \cite{st} and \cite{th}. Coalescence theory may also 
prove to be a powerful tool in reconstruction of pedigrees, as observed in \cite{wklr}.

In the present work, we assume that the known information is of a genomic type and is provided through ISSR amplification for
diploid plant cultivars, which are vegetatively propagated. ISSR amplification was popalurised by \cite{wxk} and largely used in 
genetic diverstity assessment \cite{pss}. Because being vegetatively propagated, the available dataset contains both descendants 
and ancestors in the pedigree, thus both terminal and internal nodes of the graph, while most above listed methods use information 
from last generation descendants (i.e. terminals in the graph). We know the same genotypic information for each individual and we 
assume that there are no missing individuals in the set. ISSR data only allows us to know, for each gene, if a specific allele is 
present or not. In particular, in the case of presence, we do not know if this specific allele is present in both chromosomes
(i.e. at homozygotic state, and transmitted to all the descendants) or if it is present only in one of them (i.e. at heterozygotic state and
thus transmited to only half of the descendants). It actually stems as if we observed the phenotypic expression of a dominant gene
and our model can also be applied to this kind of situation (see the discussion at the end of this paper). 
Then from this partial information we  provide a probabilistic method to reconstruct the pedigree corresponding to some families of 
diploid plant cultivars. This method consists in determining for each individual what is the most likely couple of parent pair amongst 
all older individuals,  according to some probability measure. More specifically, if $g_1,\dots,g_n$ are individuals ranked in their birth 
order, then for each $i=1,\dots,n$, we are looking for a couple of individuals possibly non distinct in the set $\{g_1,\dots,g_{i-1}\}$ which 
is the most likely parent pair of $g_i$ according to some probability measure. The construction of this measure bears on the fact that 
the probability to observe the specific alleles in the child, given the status of the parents does not depend on the generation. It only 
depends on the gene frequencies which are supposed to be constant in time. In order to justify this assumption, we prove here that 
gene frequencies converge almost surely, as the number of crossbreeding increases, toward an equilibrium which satisfies the 
Hardy-Weinberg condition.

Our reconstruction method is applied to a family of 85 living accessions representing the common broom {\it Cytisus scoparius} and 
related cultivated hybrids (Cytisus x dallimorei, Cytisus x boskoopi). The latter are diploid sexed plants whose crossbreedings have 
occurred in the past 200 years from a set of founders which is to be specified by our model. For each individual, 6 markers are used 
to highlight presence or absence of a particular allele in a high number of distinct regions of the genome. These 6 markers provide a 
total of more than 420 distinct bands for these 85 accessions, and each band has been treated as present or absent for each individual.
The results of our model applied to these particular  data are described in Section \ref{re}. Section \ref{mm} is devoted to the 
presentation of the model as well as to the convergence result of gene frequencies which justifies its relevance. Then we give some conclusions in Section \ref{conc}, comparing our results to the  existing literature and highlighting some other frameworks where 
our method can be used. 

\section{Materials and Methods}\label{mm}

\subsection{Model overview}\label{model}

We represent a pedigree as a directed graph in which each vertex corresponds to an individual and each directed edge corresponds 
to a parent-child relationship, with the edge going from parent to child.  The individuals are partitioned into two sets, $F$ and $F^\complement$, referred to as the founders and the non-founders respectively. The pedigree specifies, for every non-founder 
individual, two (not necessarily  distinct) individuals which, according to some probabilistic model shortly defined, are the most likely
parents.\\

\noindent We first define the law of reproduction in the population. Let $n$ be the number of individuals, denoted $g_1,\ldots,g_n$ 
and let $m\in\mathbb{N}$ be the number of genes for which we observe the presence or absence of a specific allele. More specifically, 
when proceeding to the  ISSR amplification, for each gene, we receive from some marker, a binary response: either the 
allele is present in at least one of the two chromosomes or it is absent in both. In particular, when the allele is present, 
we do not know if it is present on the two chromosomes. Actually, it is equivalent to consider that the allele which is highlight by the 
marker is dominant and that we only observe the phenotype of the individual. For each individual $g_i$ and each gene 
$\ell\in\{1,\ldots,m\}$, let $x_\ell(g_i)\in\{0,1\}$ be the indicator of band absences (0-values) and presences (1-values) of individual 
$g_i$ obtained during the ISSR amplification process. Hence the {\it apparent genotype} of each individual $g$ will be identified to the 
element $x(g):=(x_1(g),x_2(g),\dots,x_m(g))$ of $\{0,1\}^m$. Note that the event $\{x_\ell(g)=1\}$ means "one observes the presence of the 
allele specific to gene $\ell$ in individual $g$" or equivalently "the allelic combination of gene $\ell$ in individual $g$ is 01 or 11".\\  

\noindent  Each individual $g$ has an associated date of birth, denoted $t(g)$.  We set $t(g)=0$ if the individual $g$ was obtained 
from the wild, in which case it will be considered as a founder.  Otherwise set $t(g)$ equal to the date the individual was accessioned.
We order the individuals so that for $i<j$, $t(g_i)<t(g_j)$, whenever $t(g_j)>0$ (it is assumed that dates of birth are distinct from each 
other). The basic principles of our reconstruction method are:
\begin{itemize}
   \item[$(a)$] a uniform prior on probability $(g_j,g_k)$ are the parents of individual $g_i$ over all pairs 
   $(g_j,g_k)$ with $\max(t(g_j),t(g_k))< t(g_i)$;
  \item[$(b)$] no missing individuals, that is the parents of each non-founder individual $g_i$ belong to the set
  $\{g_1,\ldots,g_n\}\setminus\{g_i\}$.
\end{itemize}

\noindent Let us denote by $\hat{g}$ and $\bar{g}$ the parents of the individual $g$. 
When they breed, the two parents $\hat{g}$ and $\bar{g}$ with respective apparent genotypes $x(\hat{g})$ and $x(\bar{g})$ will give 
birth to the individual $g$ with apparent genotype $x(g)$ according to the following rules: 
\begin{itemize}
  \item[$(c)$] independence of the coordinates of $x(g)$, that is, $\{x_\ell(g)=1\}$ and \mbox{$\{x_{\ell'}(g)=1\}$} are independent for 
  all $\ell'\neq \ell$;
  \item[$(d)$] there are constants $\delta\in(-1/2,1/2)$ and $\varepsilon\in(0,1/2)$ called the {\it errors} and for each~$\ell$, there are constants  
  $p_\ell\in(3/4,1)$ and $q_\ell\in(1/2,1)$ such 
that for each individual $g$ and  
   \begin{itemize}
    \item $\P(\{x_\ell(g)=1\}|\,\{x_\ell(\hat{g})=1\},\{x_\ell(\bar{g})=1\})=\min(p_\ell-\delta,1)$,
    \item $\P(\{x_\ell(g)=1\}|\,\{x_\ell(\hat{g})=0\},\{x_\ell(\bar{g})=1\})=\min(q_\ell-\delta,1)$,
    \item $\P(\{x_\ell(g)=1\}|\,\{x_\ell(\hat{g})=0\},\{x_\ell(\bar{g})=0\})=\varepsilon$\,.
   \end{itemize}
\end{itemize}
Principles $(a)$ and $(b)$ should rather be considered as the most natural assumptions in the absence of any particular
constraint in the evolution of the population. Note that according to $(a)$, the father and mother can be the same individual, which is 
standard in plant populations. Principle $(c)$ means that the evolutions of different genes are independent 
between each other. In our specific example we will select a particular set of genes whose independence will be checked by means
of a statistical test, see Section \ref{re}.\\

Let us now concentrate ourself on principle $(d)$.
Constants $\delta$ and $\varepsilon$ are actually experimental errors, so they do not depend on gene $\ell$. It appears that 
when the parents satisfy $\{x_\ell(\hat{g})=1\},\{x_\ell(\bar{g})=1\}$ 
(resp. $\{x_\ell(\hat{g})=0\},\{x_\ell(\bar{g})=1\}$), the probability to observe $\{x_\ell(g)=1\}$ for the child is 
less than the theoretical probability $p_\ell$ (resp. $q_\ell$), that is $p_\ell-\delta$ (resp. $q_\ell-\delta$). Similarly, it can happen that 
when the parents satisfy $\{x_\ell(\hat{g})=0\},\{x_\ell(\bar{g})=0\}$ one observes  $\{x_\ell(g)=1\}$ for the child. This defines error
$\varepsilon$. As showed hereafter, we have $p_\ell\in(3/4,1)$ and $q_\ell\in(1/2,1)$, and the estimation from our data, see Section
\ref{re}, shows that $\delta$ and $\varepsilon$ are actually of order 0.1.\\

Besides, we recall that despite the reproduction is sexed, since we are concerned with plant populations, each individual can either 
be male or female, so that when referring to the parents  $g_j$ and $g_k$ of the individual $g_i$, the mother 
and the father are not distinguished. In particular we have 
$\P(\{x_\ell(g)=1\}|\,\{x_\ell(\hat{g})=0\},\{x_\ell(\bar{g})=1\})=\P(\{x_\ell(g)=1\}|\,\{x_\ell(\hat{g})=1\},\{x_\ell(\bar{g})=0\})$.\\

\noindent We now focus on the computation of the conditional probabilities appearing in $(d)$. In order to compute the theoretical 
values $p_\ell$ and $q_\ell$, let us assume that there is no experimental error, i.e.~$\delta=\eps=0$, so that expressions in $(d)$ 
are $\P(\{x_\ell(g)=1\}|\,\{x_\ell(\hat{g})=1\},\{x_\ell(\bar{g})=1\})=p_\ell$ and
 $\P(\{x_\ell(g)=1\}|\,\{x_\ell(\hat{g})=0\},\{x_\ell(\bar{g})=1\})=q_\ell$. Let us now compute $p_\ell$ and $q_\ell$ in terms of the
gene frequencies. We will prove in the next section that for each gene, the frequencies of the three genotypes $00$, $01$ 
and $11$,  converge toward  some equilibrium, as the number of crossbreeding increases. Let us denote respectively by 
$\pi_{00}(\ell)$, $\pi_{01}(\ell)$  and $\pi_{11}(\ell)$ these frequencies. Then in our model, we assume that this equilibrium 
is attained, so that:
\begin{itemize}
\item[$(e)$] $\pi_{00}(\ell)$, $\pi_{01}(\ell)$ and $\pi_{11}(\ell)$ do not depend on time.
\end{itemize}
Note that here, by time, we mean a scale which is incremented by successive crossbreedings. Assumption $(e)$ will be justified 
in the next section. When no confusion is possible, 
we will forget about the index $\ell$ in $\pi_{00}(\ell)$, $\pi_{01}(\ell)$ and $\pi_{11}(\ell)$.
Let us compute $p_\ell$ and $q_\ell$ in terms of $\pi_{00}$, $\pi_{11}$ and $\pi_{01}$. For a pair of parents $(\hat{g},\bar{g})$  
chosen uniformly at random in the sub-population $\{g':t(g')<t(g)\}$, the probability
to observe $x_\ell(\hat{g})=1$ and $x_\ell(\bar{g})=1$  is 
\[\P(\{x_\ell(\hat{g})=1\},\{x_\ell(\bar{g})=1\})=\pi_{11}^2+2\pi_{01}\pi_{11}+\pi_{01}^2\,.\]
When they breed and give a child $g$, the probability to observe $x_\ell(g)=1$, $x_\ell(\hat{g})=1$ and $x_\ell(\bar{g})=1$ is 
\[\P(\{x_\ell(g)=1\},\{x_\ell(\hat{g})=1\},\{x_\ell(\bar{g})=1\})=\pi_{11}^2+2\pi_{01}\pi_{11}+3\pi_{01}^2/4\,.\]
We obtain that at any time, $p_\ell$ is given by
\[p_\ell=\frac{\pi_{11}^2+2\pi_{01}\pi_{11}+3\pi_{01}^2/4}{\pi_{11}^2+2\pi_{01}\pi_{11}+\pi_{01}^2}=1-
\frac{\pi_{01}^2}{4(\pi_{01}+\pi_{11})^2}\,.\]
Then $q_\ell$ is obtained in the same way:
\begin{eqnarray*}
&&q_\ell=\frac{\pi_{01}+2\pi_{11}}{2\pi_{01}+2\pi_{11}}\,.
\end{eqnarray*}
The frequencies $\pi_{00}$, $\pi_{01}$ and $\pi_{11}$ belonging to $(0,1)$ it is easy to check from the above
expressions that $p_\ell\in(3/4,1)$ and $q_\ell\in(1/2,1)$. Furthermore, we have the relationship $p_\ell=q_\ell(2-q_\ell)$. 
In Theorem \ref{T:main}, we show that in fact the triplet of gene frequencies $(\pi_{00},\pi_{01},\pi_{11})$ satisfies the 
Hardy-Weinberg equilibrium, that is $\pi_{01}=2\sqrt{\pi_{00}\pi_{11}}$ and using this relation, we deduce that 
\begin{equation}\label{pandq}
q_\ell=\frac1{1+\sqrt\pi_{00}},\qquad p_\ell=\frac{1+2\sqrt\pi_{00}}{(1+\sqrt\pi_{00})^2}.
\end{equation}

We shall now define the set of probability measures $\mu$ from which the most likely pedigree will be derived. This definition
is based on the conditional probabilities:
\[\P(x(g)=a\,|\,x(\hat{g})=\hat{a},\,x(\bar{g})=\bar{a})=\prod_{\ell=1}^m \p(x_\ell(g)=a_\ell\,|\,x_\ell(\hat{g})=\hat{a}_\ell,
\,x_\ell(\bar{g})=\bar{a}_\ell)\,,\]
which are obtained from all acceptable triplets of individuals $(g,\hat{g},\bar{g})$ and their apparent genotypes 
$a=(a_1,\dots,a_m)$,  $\hat{a}=(\hat{a}_1,\dots,\hat{a}_m)$ and $\bar{a}=(\bar{a}_1,\dots,\bar{a}_m)$ in $\{0,1\}^m$.
More specifically, the set of individuals $\{g_1,\dots,g_n\}$ and their apparent genotype being given, for all triples 
$(i,j,k)\in\{1,\ldots,n\}^3$ and for each gene $\ell$, we first define the agreements/disagreements indicators
between the genotype of an individual $g_i$ and this of the possible couple of parents $(g_j,g_k)$:
\begin{eqnarray*}
&&p_{ijk}^{(\ell)}={\bf 1}_{\{x_\ell(g_j)=x_\ell(g_k)=x_\ell(g_i)=1\}}\,,\;\;\;\bar{p}_{ijk}^{(\ell)}={\bf 1}_{\{x_\ell(g_j)=
x_\ell(g_k)=1\,,\,x_\ell(g_i)=0\}}\,,\\
&&q_{ijk}^{(\ell)}={\bf 1}_{\{x_\ell(g_j)\neq x_\ell(g_k)\,,\,x_\ell(g_i)=1\}}\,,\;\;\;\bar{q}_{ijk}^{(\ell)}={\bf 1}_{\{x_\ell(g_j)
\neq x_\ell(g_k)\,,\,x_\ell(g_i)=0\}}\,,\\
&&\varepsilon_{ijk}=\sum_{\ell=1}^m {\bf 1}_{\{x_\ell(g_j)=x_\ell(g_k)=0\,,\,x_\ell(g_i)=1\}}\,,\;\;\;    
\bar{\varepsilon}_{ijk}=\sum_{\ell=1}^m {\bf 1}_{\{x_\ell(g_j)=x_\ell(g_k)=x_\ell(g_i)=0\}} \,.
\end{eqnarray*}
\noindent Now define $p_{\delta,\ell}=\min(p_\ell-\delta,1)$, $q_{\delta,\ell}=
\min(q_\ell-\delta,1)$,  $\bar{p}_{\delta,\ell}=1-p_{\delta,\ell}$, $\bar{q}_{\delta,\ell}=1-q_{\delta,\ell}$, $\bar{\varepsilon}=
1-\varepsilon$ and
\[
\nu_i(j,k)=\left\{
  \begin{array}{ll}
   \varepsilon^{\varepsilon_{ijk}}\cdot\bar{\varepsilon}^{\bar{\varepsilon}_{ijk}}
   \prod_{\ell=1}^m p_{\delta,\ell}^{p_{ijk}^{(\ell)}}\cdot\bar{p}_{\delta,\ell}^{\bar{p}_{ijk}^{(\ell)}}
   \cdot q_{\delta,\ell}^{q_{ijk}^{(\ell)}}
   \cdot\bar{q}_{\delta,\ell}^{\bar{q}_{ijk}^{(\ell)}}
   \,,\;\;\hbox{if $j\le k<i$\,,} \\
    0\,, \qquad \hbox{otherwise.}
  \end{array}
\right.
\]
Then for each $i=2,\dots,n$, the probability measure $\mu_i$ on $\{1,\dots,n\}^2$ is explicitly defined in terms of $x$ by
\[\mu_i(j,k)=\frac{\nu_i(j,k)}{z_i}\,,\;\;j,k\in\{1,\dots,n\}\,,\]
where $z_i:=\sum_{j,k}\nu_i(j,k)$ is a normalising constant. We readily check that $z_i>0$ for all $i$ such that $t(g_i)>0$. 
Moreover, individuals $g_i$ such that $t(g_i)=0$ are necessarily founders (i.e. $g_i\in F$), hence their parents do not belong to 
the current pedigree, so in this case, we set
\[\mu_1(j,k)=0\,,\;\;j,k\in\{1,\dots,n\}\,.\]

\noindent Fix a \emph{threshold probability} $p\in(0,1)$. Then an individual $g_i$ is in the set $F^\complement$ of non founder 
individuals, only if there exists a 
pair $(j,k)\in\{1,\ldots,n\}^2$ such that $\mu_i(j,k)\ge p$ with $j\le k<i$ (it follows that the partitioning depends on the value of $p$).\\

\noindent For each individual $g_i\in F^\complement$, we wish to determine $g_j$ and $g_k$ (possibly equal), such that the following two conditions are satisfied:
\begin{enumerate}
  \item $j\le k<i$ ($g_j$ and $g_k$ accessioned before $g_i$);
  \item $\mu_i(j,k)=\max_{j',k'}\{\mu_i(j',k'):\,j'\le k'<i\}$ ($g_j$ and $g_k$ maximize the likelihood).
\end{enumerate}
We remark that by definition of $F^\complement$, it follows that if we have found such a pair $g_j$ and $g_k$, then 
$\mu_i(j,k)\ge p$.

Note also that the normalization of the probability measure $\mu$ is relevant only for the comparison with the threshold probability. 
Steps 1.~and 2.~define the algorithm from which we performed the program in $R$ which provides the reconstructions of 
pedigrees, see Section \ref{re}.  

\subsection{Convergence to equilibrium}\label{conv}

In this subsection, we are interested in the dynamics of the frequencies of each genotype in the population. As already mentioned 
in the previous section, our reconstruction method strongly bears on the assumption that the frequencies $\pi_{00}$, $\pi_{01}$ 
and $\pi_{11}$ of the types $00$, $01$ and $11$ do not depend on time, that  is condition $(e)$ in subsection \ref{model}. We will
show in the present subsection that as the number of crossbreeding goes on, these frequencies converge almost surely to some 
random equilibrium. This result actually justifies assumption $(e)$.\\

From time $n=0$, we rank the crossbreedings in increasing order as they occur. Since the evolutions of genes are independent
of each other, see assumption $(c)$, we only need to consider the dynamics of the frequencies of genotypes $00$, $01$, $11$ 
for one gene. Then let us denote by $\pi_{00}^n$, $\pi_{01}^n$ and $\pi_{11}^n$, the proportion of individuals $g$ with 
genotype $00$, $01$ or $11$ respectively, after the $n$-th crossbreeding. Let us assume that  we start at time $n=0$
with two founders, so that after the $n$-th crossbreeding, $n+2$ individuals are present in the population. That assumes in 
particular that there is no death. Moreover we assume that both alleles exist in the two founders. Then our reproduction law 
described in $(a)$-$(d)$ of the previous subsection may actually be represented as a generalized urn model in which the probability 
of replacement depends on the proportion of individuals in the population, see \cite{pe} and the references theirin. More specifically,
at each step $n$ (crossbreeding), condition $(a)$ tells us that we choose two individuals uniformly at random in the population.\\

Let us define the polynomial function $F:\{(x,y,z)\in[0,1]^3:x+y+z=1\}\rightarrow\mathbb{R}^3$ by
\[F(x,y,z)+(x,y,z)=(xy+x^2+y^2/4,xy+yz+2xz+y^2/2,yz+z^2+y^2/4)\,,\]
and denote by $\mathcal{S}=\{(x,y,z)\in[0,1]^3:F(x,y,z)=0\}$ the zero set of $F$.

We construct $\pi^n$ recursively. Write $F=(F_1,F_2,F_3)$. At each step $n$, two uniformly chosen individuals from the population
 breed and the new frequencies of individuals with types 00, 01 and 11 become:
\begin{eqnarray*}
&&\left\{\begin{array}{ll}
&\pi_{00}^{n+1}=\frac{(n+2)\pi_{00}^n+1}{n+3}\\
&\pi_{01}^{n+1}=\frac{(n+2)\pi_{01}^n}{n+3}\,,\\
&\pi_{11}^{n+1}=\frac{(n+2)\pi_{11}^n}{n+3}\end{array}\right.
\mbox{with probability $\pi_{00}^n\pi_{01}^n+(\pi_{00}^n)^2+(\pi_{01}^n)^2/4=F_1(\pi^n)$},\\
&&\left\{\begin{array}{ll}
&\pi_{00}^{n+1}=\frac{(n+2)\pi_{00}^n}{n+3}\\
&\pi_{01}^{n+1}=\frac{(n+2)\pi_{01}^n+1}{n+3}\,,\\
&\pi_{11}^{n+1}=\frac{(n+2)\pi_{11}^n}{n+3}\end{array}\right.
\mbox{with probability $\pi_{00}^n\pi_{01}^n+\pi_{01}^n\pi_{11}^n+2\pi_{00}\pi_{11}+(\pi_{01}^n)^2/2=F_2(\pi^n)$},\\
&&\left\{\begin{array}{ll}
&\pi_{00}^{n+1}=\frac{(n+2)\pi_{00}^n}{n+3}\\
&\pi_{01}^{n+1}=\frac{(n+2)\pi_{01}^n}{n+3}\,,\\
&\pi_{11}^{n+1}=\frac{(n+2)\pi_{11}^n+1}{n+3}\end{array}\right.
\mbox{with probability $\pi_{01}^n\pi_{11}^n+(\pi_{11}^n)^2+(\pi_{01}^n)^2/4=F_3(\pi^n)$}\,.\\
\end{eqnarray*}

Let us make this construction more formal. First we define a stochastic process $(\delta_n)_n$ with values in 
$\{(1,0,0),(0,1,0),(0,0,1)\}$ in such a way that the law of $\delta_{n+1}$ conditionally on $\pi^0=i_0,\ldots,\pi^n=i_n$ is $F(i_n)$.
Recall that the quantity $(n+2)\pi^n$ represents the population size at time $n$.  Then $\pi^{n+1}$ is defined by 
\[(n+3)\pi^{n+1}=(n+2)\pi^n+\delta_{n+1}\,,\;\;\;n\ge0\,.\]
Let us set 
\[\eta_n=\delta_{n+1}-F(\pi^n)\,,\]
then we readily obtain the following equality
\begin{equation}\label{6179}
\pi^{n+1}=\pi^n+\frac1{n+3}(F(\pi^n)-\pi^n+\eta_n).
\end{equation}

For $u\in[0,1]^3$, let $f_u:\mathbb{R}^+\cup\{ 0\}\to[0,1]^3$ be the solution to the ODE
\begin{align}\label{6279} \left\{\begin{array}{ll}  
\frac{d}{dt}f_u(t)&=F(f_u(t)),\quad t\ge0,\\
f_u(0)&=u.
\end{array}\right.\end{align}

The solution can be calculated explicitly and we easily check that with $f_u(t)=(x_u(t),y_u(t),z_u(t))$ and $u=(x_0,y_0,z_0)$, then
\[ \left\{\begin{array}{ll}  
&x_u(t)=\left(x_0-\frac{(2x_0+y_0)^2}{4}\right)e^{-t}+\frac{(2x_0+y_0)^2}{4}\\
&y_u(t)=-2\left(x_0-\frac{(2x_0+y_0)^2}{4}\right)e^{-t}-\frac{(2x_0+y_0)^2}{2}+2x_0+y_0\\
&z_u(t)= 1+\left(x_0-\frac{(2x_0+y_0)^2}{4}\right)e^{-t}+\frac{(2x_0+y_0)^2}{4}-2x_0-y_0.
\end{array}\right.\]

We aim to show almost-sure convergence of $\pi^n=(\pi_{00}^n,\pi_{01}^n,\pi_{11}^n)$ as $n\to\infty$. The first step in achieving 
this is to show almost-sure convergence of $v(\pi^n)$ as $n\to\infty$, where $v(u):=\lim_{t\to\infty}f_u(t)$. This is achieved in the 
following lemma.

\begin{lemma}
As $n\to\infty$, $v(\pi^n)$ converges almost surely.
\end{lemma}
\begin{proof}
We shall show that almost surely, $(v(\pi^n))_n$ is a Cauchy sequence. We have
\begin{align}\label{Eq:v1}
|v(\pi^{n+1})-v(\pi^n)|\le\left|v\left(\pi^n+\frac1{n+3}F(\pi^n)\right)-v(\pi^n)\right|+\left|v(\pi^{n+1})-v\left(\pi^n+\frac1{n+3}F(\pi^n)\right)\right|.
\end{align}
We provide upper bounds on each term appearing on the right-hand side. Firstly, using the fact that $v(x)=v(f_x(t))$ for any $t\ge0$,
\[
\left|v\left(\pi^n+\frac1{n+3}F(\pi^n)\right)-v(\pi^n)\right|=
\left|v\left(\pi^n+\frac1{n+3}F(\pi^n)\right)-v\left(f_{\pi^n}\left(\frac1{n+3}\right)\right)\right|.
\]
We have the explicit form of $v$ as 
\[
v(u)=\left(\frac{(2x_0+y_0)^2}{4},-\frac{(2x_0+y_0)^2}{2}+2x_0+y_0,1+ \frac{(2x_0+y_0)^2}{4}-2x_0-y_0\right),
\]for any $u=(x_0,y_0,z_0)$. The function $v$ is clearly Lipschitz on $[0,1]^3$ and so there exists a constant $c$ such that 
\begin{align*}
\left|v\left(\pi^n+\frac1{n+3}F(\pi^n)\right)-v\left(f_{\pi^n}\left(\frac1{n+3}\right)\right)\right|&\le c\left|\pi^n+\frac1{n+3}
F(\pi^n)-f_{\pi^n}\left(\frac1{n+3}\right)\right|\\&\le O(1/n^2),
\end{align*}
since $f_{\pi^n}(1/(n+3))=f_{\pi^n}(0)+\frac1{n+3}f'_{\pi^n}(0)+O(1/n^2)=\pi^n+\frac1{n+3}F(\pi^n)+O(1/n^2)$. For the second term 
on the right-hand side of (\ref{Eq:v1}), we have 
\[
\left|v(\pi^{n+1})-v\left(\pi^n+\frac1{n+3}F(\pi^n)\right)\right|\le 
c\left|\pi^{n+1}-\pi^n-\frac1{n+3}F(\pi^n)\right|\le\frac{c}{n+3}|\eta_n-\pi^n|,
\]
by the definition of $\pi^n$, see (\ref{6179}). However since $F$ is bounded we deduce that we can upper bound this term by $O(1/n)$. 
Plugging the two bounds we have obtained into equation (\ref{Eq:v1}) shows that the sequence $(v(\pi^n))_n$ is indeed Cauchy (surely), 
and this completes the proof.
\end{proof}

We are now in a position to show almost-sure convergence of the stochastic process $\pi^n=(\pi_{00}^n,\pi_{01}^n,\pi_{11}^n)$, $n\ge1$.
 
\begin{thm}\label{T:main} The random vector $\pi^n=(\pi_{00}^n,\pi_{01}^n,\pi_{11}^n)$, $n\ge1$ has the following asymptotic behaviour:
\[\pi^n\stackrel{\mbox{\it\footnotesize a.s.}}{\longrightarrow}(\pi_{00},\pi_{01},\pi_{11})\,,\;\;\mbox{as $n$ tends to $+\infty$}\,,\]
where $(\pi_{00},\pi_{01},\pi_{11})$ is distributed on $\mathcal{S}$. In particular, it satisfies the 
Hardy-Weinberg equilibrium:
\[\pi_{01}=2\sqrt{\pi_{00}\pi_{11}}\,. \]
\end{thm}
\begin{proof}
We first claim that almost surely, the $L^1$ distance between $\pi^n$ and $\mathcal{S}$ tends to 0 as $n\to\infty$. Recall that the $L^1$ 
distance $|\pi^n-\mathcal{S}|$ is defined as 
\[
|\pi^n-\mathcal{S}|:=\min_{s\in\mathcal{S}}\{|\pi^n-s|\}:=\min_{(x,y,z)\in\mathcal{S}}\{|\pi_{00}^n-x|+|\pi_{01}^n-y|+|\pi_{11}^n-z|\}
.\] In fact, this is a consequence of Theorem 2.2 in \cite{sc} which asserts that the limit set of $(\pi^n)$ (i.e.  the set of limits of 
subsequences 
of $(\pi^n)$) is almost surely a connected compact internally chain recurrent set for the flow associated to the 
ODE (\ref{6279}). In particular the limit set of $(\pi^n)$ is included in $\mathcal{S}$, which implies that the distance between 
$\pi^n$ and $\mathcal{S}$ tends almost surely to 0.

Suppose $x\in\mathcal{S}$ so that $F(x)=0$ by definition. Then $\frac{d}{dt}f_x(t)=0$ for all $t\ge0$ and so $f_x(t)=x$ for all 
$t\ge0$, and in particular $v(x)=x$. Since $v$ is Lipschitz and $v(\mathcal{S})=\mathcal{S}$ we have that, as $x\to\mathcal{S}$, 
$|v(x)-x|\to0$. 
But since $v(\pi^n)$ converges almost surely to some limit random variable, we deduce that $\pi_n$ also converges almost 
surely and to the same limiting random variable.

Finally, Hardy-Weinberg equilibrium follows readily from the fact that $(\pi_{00},\pi_{01},\pi_{11})$ is distributed on the set $\mathcal{S}$,
i.e. $F(\pi_{00},\pi_{01},\pi_{11})=0$.
\end{proof}

\noindent In this theorem, an additional information is brought  by the Hardy-Weinberg principle which provides a 
relationship between  the allelic frequencies and the genotypic frequencies. This equilibrium was predictable and is actually 
a natural consequence of the absence of any evolutive forces. \\

Let us now consider the general case $m\ge1$. We denote by  $\pi_G$ the frequency of a genotype 
$G=(G_1,\dots,G_m)\in\{00,01,11\}^m$. If $\pi_{i,00}$, $\pi_{i,01}$ and  $\pi_{i,11}$, are respectively the limiting 
gene frequencies of the $i$-th gene with alleles $0$ and $1$, then from the independence between genes 
(see condition $(c)$ in the previous subsection), the limiting frequency of the 
genotype $G$ at equilibrium is
\[\pi_G=\pi_{1,G_1}\pi_{2,G_2}\dots\pi_{m,G_m}\,.\]

\begin{remark} It is a quite challenging question to determine the exact distribution of the limit triplet $(\pi_{00},\pi_{01},\pi_{11})$. 
Actually our simulations show that it may have a diffuse distribution in the set  $\{(x,y,z)\in[0,1]^3:x+y+z=1\}$, which depends 
on the initial values $\pi_{00}^0$, $\pi_{01}^0$ and $\pi_{11}^0$, see Figure $\ref{dfs}$.
\end{remark}

\begin{figure} 
\includegraphics[width=15cm]{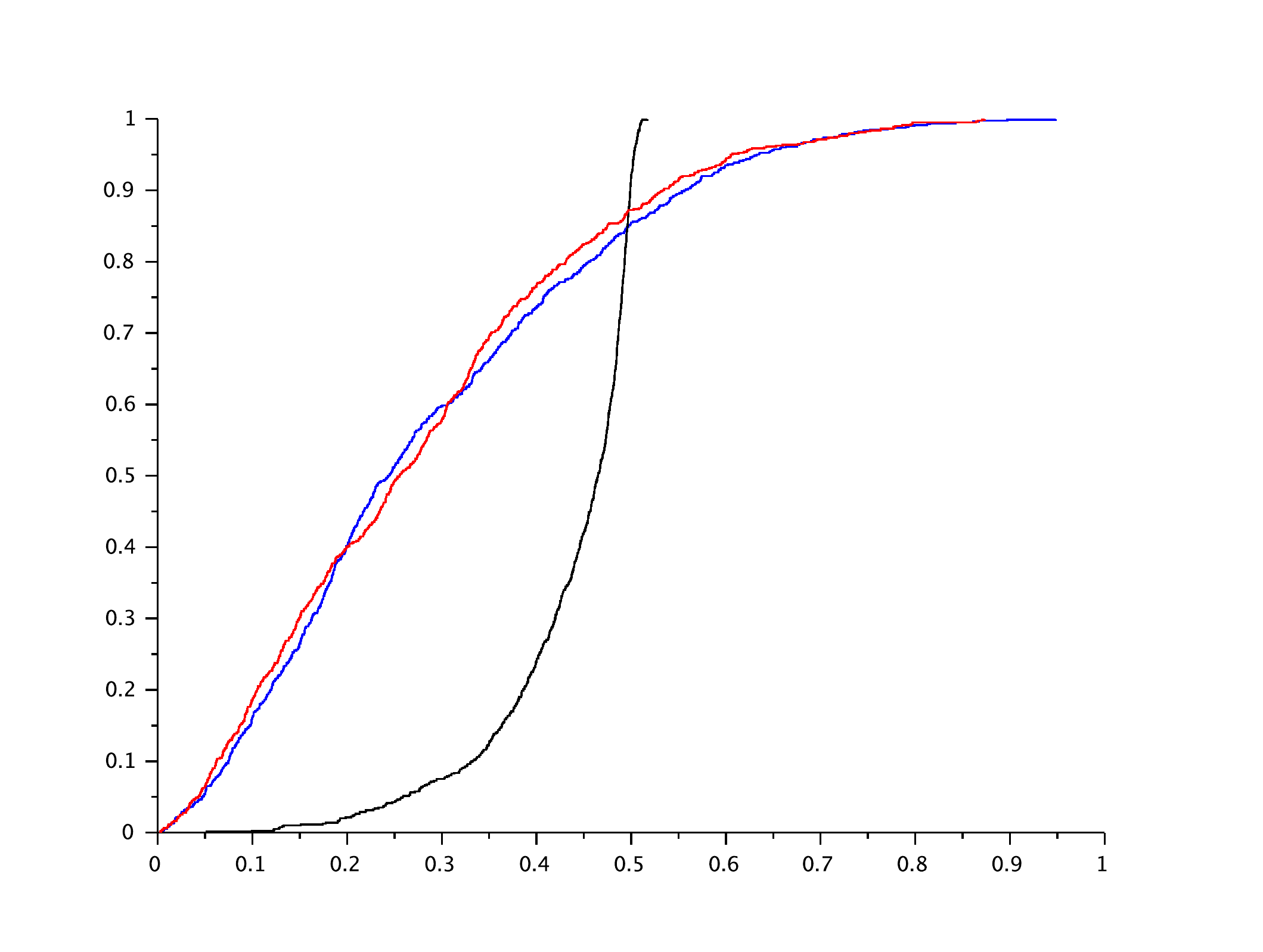}
\includegraphics[width=15cm]{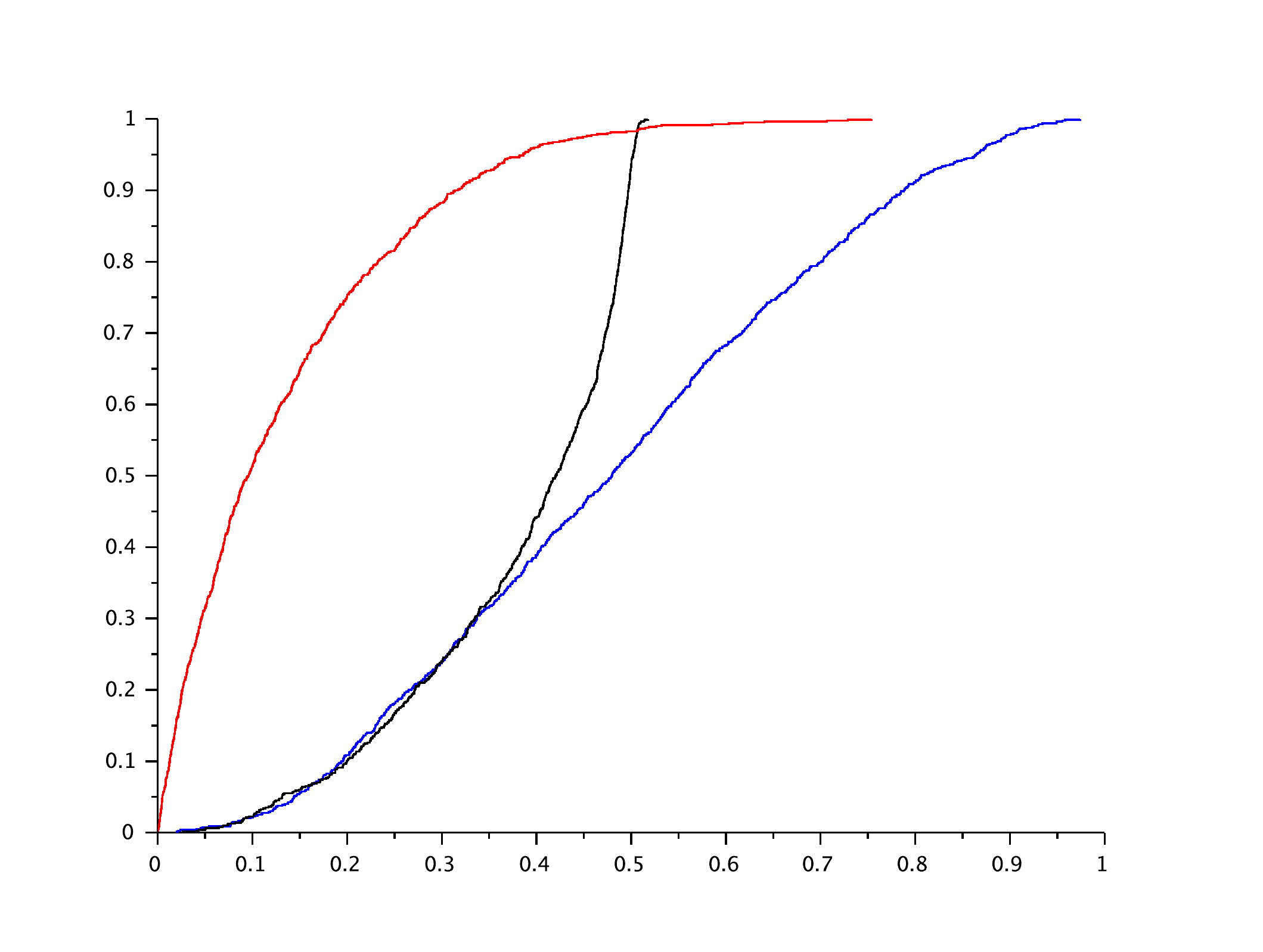}
\caption{Empirical distribution functions of $\pi_{00}$ (blue), $\pi_{11}$ (red) and $\pi_{01}$ (black). The first 
figure is obtained with initial values $\pi_{00}^0=1$, $\pi_{01}^0=2$, $\pi_{11}^{0}=3$ and the second one
is obtained with $\pi_{00}^0=1$, $\pi_{01}^0=1$, $\pi_{11}^{0}=0$.}\label{dfs}
\end{figure}

\begin{remark}
A subsequent question to Theorem \ref{T:main} concerns the speed of convergence of $(\pi_{00}^n,\pi_{01}^n,\pi_{11}^n)$. 
Some results in this direction are given in $\cite{de}$ and $\cite{hmps}$. However, they require some strong
assumptions on the derivative of the function $F$ at the limiting point $(\pi_{00},\pi_{01},\pi_{11})$, which 
are quite difficult to verify in our situation, mainly due to the fact that we do not know the distribution of  
$(\pi_{00},\pi_{01},\pi_{11})$. However, it is reasonable to expect that a central limit type  theorem holds, 
in which case, the speed of convergence of $(\pi_{00}^n,\pi_{01}^n,\pi_{11}^n)$ to $(\pi_{00},\pi_{01},\pi_{11})$ 
would be of order $\sqrt{n}$.
\end{remark}

\section{Application of the model}\label{re}

Our model were tested on a population of 85 living accessions representing the common broom 
\emph{Cytisus scoparius} and three related interspecific hybrids.  This dataset consists in 62 vegetatively propagated
cultivars obtained from various nurseries. These cultivars belong to either {\it Cytisus scoparius}, {\it Cytisus} x 
{\it dallimorei} (hybrid between {\it C. scoparius} and {\it C. multiflorus}), {\it C.} x {\it praecox} (hybrid between {\it C. multiflorus} 
and {\it C. oromediterraneus}), or {\it C.} x {\it booskopii} (hybrid between {\it C.} x {\it dallimorei} and {\it C.} x {\it praecox}). 
In addition three to nine individuals obtained from five wild populations have been included (3 individuals of {\it Cytisus 
oromediterraneus} from France, 3 individuals of {\it Cytisus scoparius} from Italia, 3 from Poland, 4 from Angers, France and 9 
from Ern\'ee, France). For all these samples, DNA extration use the Nucleospin\textregistered Plant II kit 
from macherey-Nagel. IISR data was obtained using six set 
of primers, namely ISSR5 (sequence: 5¡¯-CACACACACACACACARC-3¡¯), ISSR7 (sequence : 
5¡¯-CACACACACACACACART-3¡¯), ISSR13 (sequence:\\ 5¡¯-GTGTGTGTGTGTGTGTYA-3¡¯), 
ISSR890 (sequence: 5¡¯-VHVGTGTGTGTGTGTGT-3¡¯), ISSR891 (sequence :
5¡¯-HVHTGTGTGTGTGTGTG-3¡¯) and ISSRa (sequence: 5¡¯-GCTCTCTCTCTCTCTC-3¡¯).
Polymerase chain reaction (PCR) was done using the following parameters : $95^o$C for 2 min., then 
39 cycles of $95^o$C for 30 sec., $50^o$C for 30 sec., $72^o$C for 120 sec., followed by 10 min. of
extension at $72^o$C. Electrophoresis was done on 5\% acrylamide-bisacrylamide gel (mixing ratio : 29:1), 
with 7M urea, with a pre-run of 30 min at 80 W, then 2h30 at 60W. Staining use silver nitrate. Gels  
were scanned and band manualy read.

Using data obtained from ISSR analysis, our present aim is to determine the most likely pedigree relating 
these individuals. A code in language $R$ has been written according to the model described in the previous
sections. The latter applied to our data provided the pedigrees presented in figures \ref{fig1},  \ref{fig2} and \ref{fig3} 
below. The use of this method first requires that the population we are dealing with satisfies principles $(a)-(e)$ in 
Subsection \ref{model} and parameters $\varepsilon$, $\delta$, $p_\ell$ and $q_\ell$ must be inferred from our data. 

Breedings have occurred over time under the action of professional breeders or
according to natural phenomenons and with no more information, assumption $(a)$ about uniform 
prior distribution is reasonable.  According to botanists, this is also the case of assumption $(b)$ which 
means that there are no missing individuals in the population. Then we need to ensure  
the independence hypothesis (c) between the bands $\{x_\ell(g)=1\}$, $\ell\in\{1,\dots,m\}$. Depence may 
occur due to the selective sweep phenomenon which can associate together several genes whose loci are close 
to each other along the chromosome. For such sets of genes, recombination is not strong enough for them to be 
considered as independent in the reproduction process.  Then among the 424 bands, we have selected 168 of them 
which are proved to be independent from a statistical test.

We also need to determine the values of $\varepsilon$, $\delta$, $p_\ell$ and $q_\ell$ related to the present 
data, in order to construct the probability measure which is defined in $(d)$. First recall that in
the ISSR amplification, six markers allow us to test the presence or absence of those 168 bands, 
each marker corresponding to a particular set of bands (34 bands for ISSR890, 22 for ISSR 891, 31
for ISSRa, 32 for ISSR5, 27 for ISSR7 and 22 for ISSR13). For each of the six markers used, in order to apply the above 
model, we need to estimate the values of $\delta$ and $\eps$ (the errors probability, which can occur during the experiment). 
We achieve this by repeatedly crossing two individuals (G017 \emph{Cytisus scoparius 'Lunagold'} and G010 
\emph{Cytisus x dallimorei 'Burkwoodii'}) and 
performing marker analysis  (using 5 of the 6 markers used for the dataset) on the resulting offspring  (n=33 plants). We are 
then able to estimate, for each marker, the value of $\delta$. Denoting by $\delta_m$ the error using marker $m$, we assume 
that $\delta_m$ is a Gaussian random variable such that  $\var(\delta_m)=\var(\delta_{m'})$ for all markers $m,\,m'$.  We 
obtained the following average errors: 
\begin{eqnarray*}
&&\E(\delta_{ISSRa})=0.16,\, \E(\delta_{ISSR890})=0.16,\,\\
&&\E(\delta_{ISSR891})=0.14,\, \E(\delta_{ISSR5})=0.19,\,\E(\delta_{ISSR7})=0.1.
\end{eqnarray*}
For each pair of markers, $m$ and $m'$, we ran a hypothesis test to determine whether $\E(\delta_m)=\E(\delta_{m'})$ 
and we found that we do not reject this null hypothesis at a 95\% confidence level. We obtained a 95\% confidence interval 
of $(0.126,0.195)$ for the error, under the assumption that the errors from the different  markers all came from the same 
distribution. For the present reconstructions we have chosen the value $\delta=0.15$. The same study for the error $\eps$ 
leads us to the choice of $\eps=0.05$.\\ 

In subsection \ref{conv} we proved convergence of gene frequencies and we will assume that the population which 
is considered here has attained some equilibrium, that is principle $(e)$.  As can be seen from equation (\ref{pandq}), 
thanks to Hardy-Weinberg  principle, the probabilities $p_\ell$ and $q_\ell$ only depend on the probability $\pi_{00}$. 
We emphasize that the latter probability is actually the only one whose empirical value can be determined from the 
data. Indeed it is not possible to distinguish the genotype $01$ from the genotype $11$ in ISSR data. In the present case, we 
obtain the values of $\pi_{00}$ and hence $p_\ell$ and $q_\ell$ for each band.\\

The probabilities $\mu_i(j,k)$ defined in the end of Subsection \ref{model} may appear quite low once computed 
from our dataset. However knowing that all individuals belong to the same family, we are only concerned with their 
relative values. The pedigrees appearing in figures \ref{fig1}, \ref{fig2} and \ref{fig3} were obtained with the threshold 
probabilities $0.1$ and $0.2$ and $0.3$ respectively. Funders have been represented in black and individuals with no 
parent and children have not been represented.  As expected, when the threshold probability $p$ increases, 
the number of relations between individuals decreases  and more  
individuals are considered as founders. Compared to the existing knowledge we have on the 
group (see \cite{au}), several relationships are congruent with historical information. 
For example, 'Zeelandia' is reported as a descendant of 'Burkwoodii' and a C. x praecox. This relationhip appears with 
all threshold probabilities. 'Liza', 'Andreanus Select', and 'Donard Seedling' are all all historically reported as sport 
(bud mutations) of 'Burkwoodii', while 'Lena' is supposed to be a seedling of it. They are all linked under $p=0.1$ and $p=0.2$, 
while under higher threshold probability 'Burkwoodii', 'Liza' and 'Andreanus Select' are still linked, however, Donard Seedling is 
treated as a seedling of 'Burkwoodii' and Cytisus ardoinoi which may be impossible (the sample used for representing this last 
species being wild collected). 'Firefly is reported as a seedling of 'Andreanus', which appears under all threshold probabilities. 
Comparing to historical information, 'La Coquette' appears here as founder, and as parent of 'Roter Favorit' while it was reported 
as a self-fecondation of 'Hollandia', and half-brother of 'Boskoop Ruby'. 'Hollandia' is know to be a seedling from 'Burkwoodii' and 
C. x praecox, here, under p=0.1, it is a seedling between the same 'Burkwoodii' but with C. scoparius. Using the same ISSR data, 
Auvray in \cite{au} points out the putative link between 'Apricot Gem' and 'Dukaat', as well as between 'Boskoop Ruby' and
'Windlesham'. These links are re-inforced here and second putative parents are provided (kewensis for 'Apricot Gem' and 'Hollandia' 
for 'Windlesham'). Auvray [1] also point out a parentage between 'Moclard Pink' and 'Minstead' (the former being a putative seedling 
of the later), here 'Moclard Pink' is always linked with 'Albus', a point which needs consideration. Under the various threshold 
probabilities, 'Luna', 'Palette' and 'Roter Favorite' are linked, this seems reasonably consistent with the fact that they all have been 
obtained form the same nursery (Arnold, at Alreslohe near Holstein in Germany) around 1960. 'Jessica', linked to the same group 
under $p=0.1$ is of unknown parentage, while 'Goldfinch', also linked under $p=0.1$ is reported to be a seedling between 'Donard 
Seedling' and 'Dorothy Walpole' (laking from the sampling). The links between 'Andreanus', 'Firefly', 'Golden Sunlight' 'Andreanus Splendens', 'Golden Cascade', 'Roter Favorite' and 'Queen Mary', appearing under all threshold probabilities, reminds that all these 
cultivars are selection of C. scoparius and not of any of the interspecific hybrids.

\begin{figure} 
\includegraphics[width=16cm]{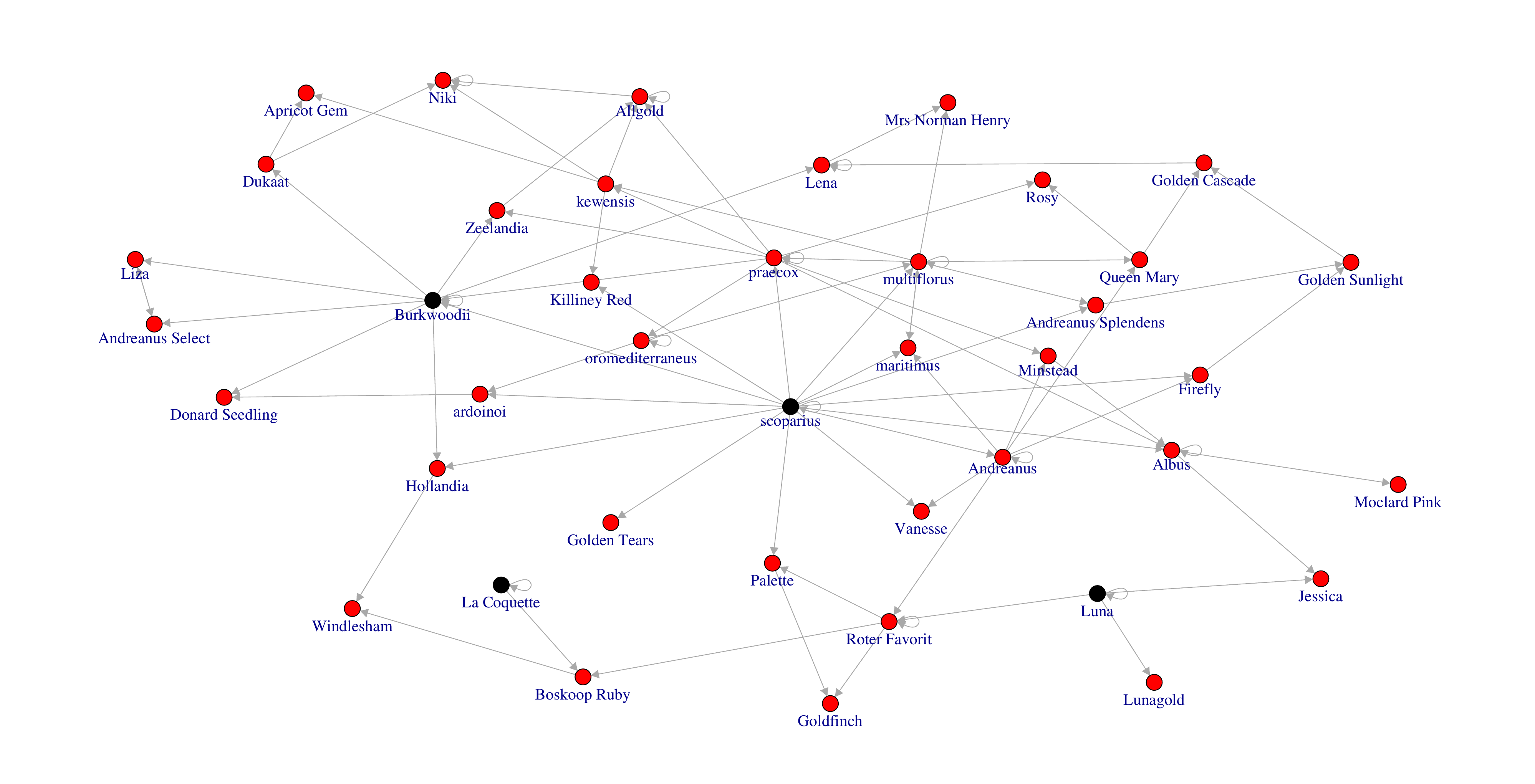}
\caption{Threshold probability $p=0.1$.}\label{fig1}
\end{figure}

\begin{figure} 
\includegraphics[width=16cm]{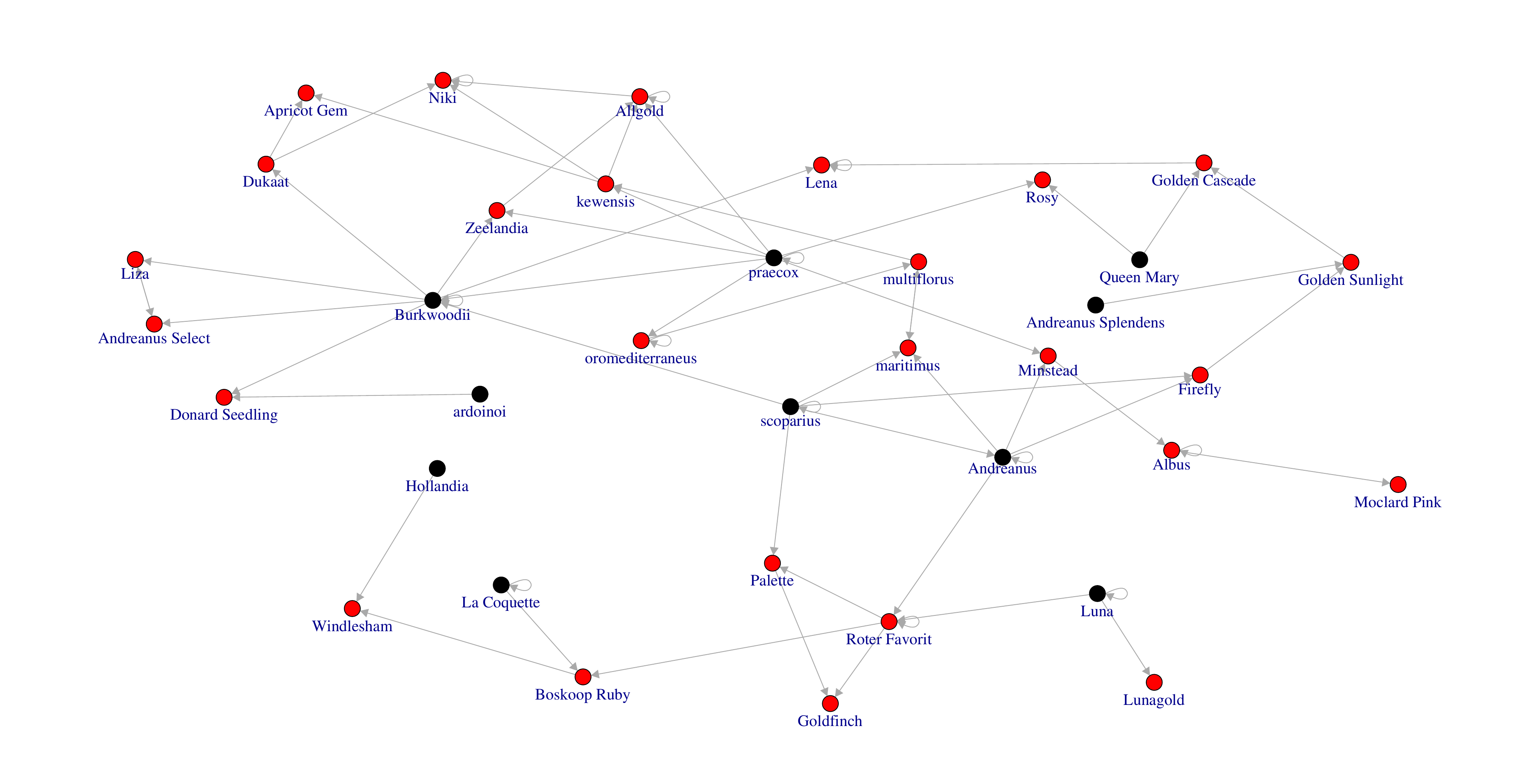}
\caption{Threshold probability $p=0.2$.}\label{fig2}
\end{figure}

\begin{figure} 
\includegraphics[width=16cm]{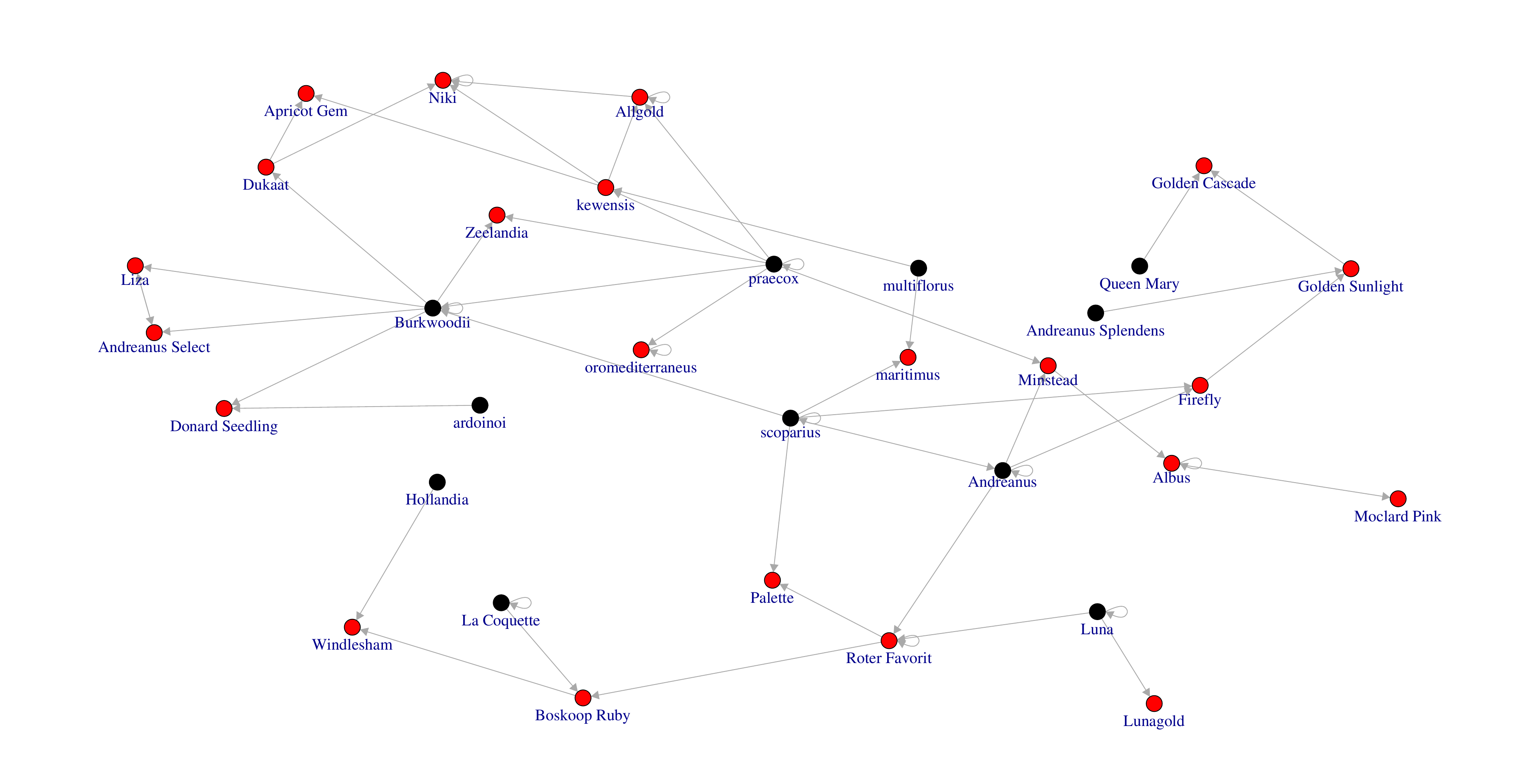}
\caption{Threshold probability $p=0.3$.}\label{fig3}
\end{figure}

\section{Discussion}\label{conc}

We have set up a mathematical model of pedigree reconstruction whose basic principle is to determine,
for each individual, what is the most likely parent pair in the population, according to the probability
distribution which is defined in $(d)$ of Subsection \ref{model}. The robustness of this model mainly
relies on the fact that gene frequencies have attained some equilibrium. We show  in Subsection
\ref{conv} that indeed, in the absence of any evolutive forces, gene frequencies converge toward a 
limit random vector which satisfies Hardy-Weinberg equilibrium. From this model we derived an
algorithm which is written in language R and then  we applied this model to ISSR data from a
population of diploid plants. The results reveal that the pedigrees obtained from this method fit to the
partial reconstructions based on botanical data or other methods using dendograms obtained from 
matrix distances. This additional source of information could also be used in order to improve the 
model by constructing a new probability distribution giving a relative weight to each kind of data. 

Greater power could also be given to our method by getting rid of assumption $(b)$ on non missing
individuals. Indeed missing individuals in the population who would actually have lots of family
relationships could considerably distort the real pedigree. Then an improvement would consist in
determining how much the addition of one or several virtual individuals with specific genomes 
increases the likelihood of the pedigree. 

Principle $(c)$ assumes that recombination is uniform, but this can be made more realistic by
determining how different sets of loci actually recombines from a preliminary statistical inference. 
Then the model can easily be adapted. 

Finally we emphasize that our model can be applied to phenotyped data. Indeed, as already
observed in Section \ref{mm}, the knowledge of ISSR is equivalent to the knowledge of the expression
of a dominant gene. Hence our model can easily be tested from a population about which we observe 
a specific set of phenotypical criteria and whose family relationship are a priori known.

\noindent {\bf Acknowledgements}
Projects EUROGENI and BRIO have been managed by V\'eronique Kapusta, while molecular and bibliographic 
information concerning Cytisus material had been acquired by Ga\"elle Auvray, Agathe Le Gloanic and 
Nad\`ege Le Pocreau. We warmly thank all of them for their help.

\end{document}